\documentclass[10pt,a4paper]{amsart}
 	\usepackage{setspace}
	\usepackage[latin1,utf8]{inputenc} 
	\usepackage[english]{babel}
	\usepackage[T1]{fontenc}
	\usepackage{amsmath,amsopn,amssymb,amsthm}
	\usepackage{ xcolor}
  \usepackage{ tikz-cd}
	\newtheorem{thm}{Theorem}
	
	\newtheorem{lem}{Lemma}
	\newtheorem{rem}{Remark}
	
	\newtheorem{cor}{Corollary}
	\newtheorem{mdef}{Definition}
	\newtheorem{prop}{Proposition}

	 \let \a = \alpha  \let\ph=\varphi    
	
	\def\e{{\rm e}}
	\def\i{{\rm i}}
	\def\eps{{\varepsilon}}
	\def\Id{{\rm Id}}
	\def\op{{\rm op}}
	\def\opaw{{\rm op}_{\rm aw}}
	\def\R{\mathbb R}
	\def\Rd{{\mathbb{R}^d}}
	\def\Rdd{{\mathbb{R}^{2d}}}
	\def\C{\mathbb C}
	\def\Cd{{\mathbb{C}^d}}
	
	\def\N{\mathbb N}

	\def\P{\mathcal P}
	\def\T{\mathcal T}
	\def\B{\mathcal B}
	\def\W{{\mathcal W}}
	
	\def\F{{\mathcal F}}
	\def\Im{{\rm Im}}
	\def\Re{{\rm Re}}
	\def\Span{{\textnormal{{Span}}}}
	
	\def\BB{{\textbf{B}}}
	\def\PP{{\textbf{P}}}
	\def\TT{{\textbf{T}}}
	\def\FF{{\textbf{F}}}

	\def\M{{\mathcal{M}}}

	\def\dd{{\rm{d}}}

	\def\({\left(} \def\){\right)}  
	   \def\lk{\,\left[ \,} \def\rk{\,\right] \,} 
	   \def\lb{\left\{} \def\rb{\right\}}  
	   \def\lw{\left\langle} \def\rw{\right\rangle}

 	\definecolor{blue}{RGB}{0,00,255}
 	\definecolor{red}{RGB}{255,30,30}

\title[Polyanalytic Toeplitz operators]{Polyanalytic Toeplitz operators: isomorphisms, symbolic calculus and approximation of Weyl operators}
\author{Johannes Keller \and Franz Luef }
\email{johannes.f.keller@gmail.com \and franz.luef@math.ntnu.no}
\address{Department of Mathematical Sciences\\ NTNU Trondheim\\7041 Trondheim, Norway}

\date{\today}


\begin{document}

\maketitle

\begin{abstract}
We discuss an extension of Toeplitz quantization based on polyanalytic functions. 
We derive isomorphism theorem for polyanalytic Toeplitz operators between
weighted Sobolev-Fock spaces of polyanalytic functions, which are images of 
modulation spaces under polyanalytic Bargmann transforms.
 This generalizes well-known results from the analytic setting. 
 Finally, we derive an asymptotic symbol calculus and present an asymptotic expansion of complex Weyl operators in terms of polyanalytic Toeplitz operators.
\end{abstract}

\section{Introduction}

Bargmann transforms and Fock spaces provide a widely used language that connects the  
theory of entire functions with a variety of topics in  theoretical and applied mathematics, including signal analysis, 
quantum mechanics as well as complex geometry and  analytic microlocal analysis.

This area of mathematics goes back to the seminal work \cite{ba61} of Bargmann that has been motivated by 
applications in quantum mechanics. In microlocal analysis, generalized Bargmann transforms are mostly known
as Fourier-Bros-Iagolitzer transforms and were first applied
by  Bros, Iagolnitzer and Stapp in order to analyze wave front sets, see e.g.~\cite{IS69}, or \cite{Sj82} for a more recent and general approach.
Janssen established the link between the Bargmann transform and Gabor frames in \cite{ja82} which allowed him to apply methods from 
complex analysis to problems in signal analysis. This connection between Gabor frames and complex analysis has turned out to be 
very fruitful, e.g. for the characterization of the Gabor frame set of a Gaussian in \cite{ly92,se92-1}, or the construction of 
unconditional bases for Bargmann-Fock spaces in \cite{fegrwa92}. 

Toeplitz operators provide a natural framework to describe linear transformations in Fock type 
spaces that can be interpreted as signal manipulations, quantum observables or 
pseudodifferential operators. In fact, Toeplitz operators are nothing else but the image of
 anti-Wick  or localization operators under the Bargmann transform. Putting it differently, 
 localization operators are in fact Toeplitz operators on the phase space, see~\cite{BCG04,Eng09,AF15}.

The aim of this paper is to lift the well-established theory of Toeplitz operators to the polyanalytic setting,
following initial works of Abreu, Gröchenig and Faustino~\cite{Ab10,AG10,F11} as well as \cite{Gof10,EZ17,RV19}. That is, we 
introduce multiplication operators on  Bargmann-Fock type spaces of polyanalytic functions and, thus, provide a 
whole new family of quantization schemes. Polyanalytic Toeplitz operators appear as the natural 
complexification of localization operators with Hermite function windows. Moreover, polyanalytic Bargmann-Fock 
spaces are precisely the images of the classical modulation spaces under polyanalytic Bargmann transforms.

Polyanalytic functions were first studied by Kolossov more than a century ago. Howoever, it was not until the seminal work 
of Vasilevski \cite{vasilevski2000poly} that this generalization of analytic functions has received more attention. The increasing importance 
in mathematics and signal analysis is due to the link between Gabor superframes generated by Hermite functions which are intrinsically 
related to polyanalytic spaces \cite{Ab10,grly09}. In \cite{AG10} the theory of Bargmann-Fock spaces has been extended to the setting of 
polyanalytic functions, see also \cite{AF14} for a survey on these recent developments. One of our main results is a lifting theorem for 
modulation spaces of Gr\"ochenig-Toft \cite{GT11,GT13} to polyanalytic Bargmann-Fock spaces.

Motivated by applications in analytic microlocal analysis and semiclassical quantum theory, in this paper we formulate all 
results in a semiclassical scaling by assuming that $1 \gg \hbar>0$ is a small parameter.

This paper is structured as follows: after reviewing some basics
about Bargmann transforms and quantization in~\S\ref{sec:prelim}, in \S\ref{sec:poly_Toeplitz} we introduce 
the idea of polyanalytic Bargmann transforms as well as polyanalytic Toeplitz quantization $\T_k(m)$ of a 
symbol $m:\Cd\to\C$, where $k\in\N^d$ indicates the degree of polyanalyticity.
\S\ref{sec:isom} contains our first main theorem, namely, isomorphism results of the form
\[
\T_k(m) : \F^{k,p,q}_m(\Cd) \to \F^{k,p,q}(\Cd)
\]
for polyanalytic Toeplitz operators as maps  between polyanalytic Sobolev-Fock  spaces 
$\F^{k,p,q}_m(\Cd)$. 
These spaces appear as images of the well-known modulation spaces under the polyanalytic Bargmann transform. 
In~\S\ref{sec:poly_toeplitz_calculus} we present an $\hbar$-dependent 
asymptotic symbol calculus for localization operators $\opaw^{\ph_k}(a)$, where the window $\ph_k$ 
is a Hermite function, as well as for their complex counterparts, namely, polyanalytic Toeplitz operators.
For example, we show that the commutator of two Hermite localization 
operators $\opaw^{\ph_k}(a)$ and $\opaw^{\ph_k}(b)$ has an asymptotic expansion of the form
\[
 \tfrac{\i}{\hbar}  \lk \opaw^{\ph_k}(a),\opaw^{\ph_k}(b) \rk =\opaw^{\ph_k}\(\lb a,b \rb\) + O(\hbar)
\]
with $\lb \cdot,\cdot \rb$ the usual Poisson bracket on $\Rdd$, and thus corresponds to a 
$O(\hbar)$ deformation of the classical phase space.
Finally, in \S\ref{sec:weyl_toeplitz} we apply the new concepts
to prove an asymptotic expansion of complex Weyl quantized 
operators in terms of polyanalytic Toeplitz operators.

In summary, we obtain a whole range of new and related quantization schemes whose combination allows for a refined analysis and 
more precise approximations. It is the hope of the authors that polyanalytic Toeplitz operators will prove useful in various applications 
such as manipulation of multiplexed signals, construction and analysis of Gabor superframes and semiclassical quantum dynamics.

\section{Background}\label{sec:prelim}

We start by reviewing some concepts and results that form the basis for 
the subsequent introduction and investigation of polyanalytic Toeplitz operators.
We first recall Bargmann transforms as well as the well-known Toeplitz, Weyl and anti-Wick 
quantization  schemes. 
Moreover, for the reader's convenience and later reference we recall the spectrogram expansion of 
Wigner functions  from~\cite{K17}.

\subsection{Bargmann transform}
The Bargmann transform $\B$ --- see, e.g., the standard reference~\cite[\S I.6]{F89} --- maps 
the usual Hilbert space $L^2(\Rd)$ of quantum mechanics and signal analysis 
into the Fock space
\[
\F(\Cd) := \Big\{ F:\C^d\to\C: F \text{ is entire and } \| F\|_{L^2_\Phi} < \infty\Big\}
\]
which is a closed subspace of the weighted Hilbert space 
 \[
L^2_\Phi(\Cd) :=  L^2(\C^d,\e^{-2\Phi(z)/\hbar} \dd z).
\] 
with the  strictly plurisubharmonic exponential weight function
\[
\quad\Phi(z)=\tfrac14|z|^2
\]
 and the  norm
\[
\| F\|^2_{L^2_\Phi} :=  \int_{\C^d} |F(z)|^2 \e^{-|z|^2/2\hbar} \dd z.
\]
Hence, the Fock space $\F(\Cd)$ consists of entire functions of $d$ variables with controlled 
growth behaviour at infinity.
Analoguously to \cite{Ab10},
we define the $d$-dimensional $\hbar$-rescaled Bargmann transform as
\[
\B: L^2(\Rd) \to \F(\Cd), \qquad \B\psi (z) = (2\pi\hbar)^{-d/2}(\pi \hbar)^{-d/4} \int_\Rd 
\psi(x) \e^{(xz -z^2/4 - x^2/2)/\hbar}\dd x
\]
with $\hbar >0$ a small parameter.
In the analytic microlocal setup from~\cite[\S 13]{Z12}, the operator $\B$ corresponds to
 the  Fourier-Bros-Iagolnitzer transform associated with the holomorphic quadratic phase  
 \begin{equation}\label{eq:theta_phase}
 \theta(z,w) = \tfrac{\i}2((z-w)^2-z^2/2)
 \end{equation}
and the corresponding strictly plurisubharmonic exponential weight function
\[
\quad\Phi(z)= -\max_{x\in \Rd}\, \Im \,\theta(z,x) 
\]
from above that gives rise to the Hilbert space $L^2_\Phi(\Cd)$.
The Bargmann transform $\B:L^2(\Rd) \to \F(\Cd)$ is unitary and the
associated orthogonal Bergman projector 
\begin{equation}\label{eq:bergman_proj_def}
\P := \B\B^*
\end{equation} 
maps $L^2_\Phi(\Cd)$ into its closed subspace $\F(\Cd)$. One computes its adjoint operator $\B^*$  explicitely  as
\[
\B^*F(x) =(2\pi\hbar)^{-d/2}(\pi \hbar)^{-d/4} \int_\Cd F(w) \e^{-(\overline w -x)^2/2\hbar+\overline w^2/4\hbar} 
\e^{-|w|^2/2\hbar} \dd w, \quad x\in\Rd,
\]
for any function $F\in L^2_\Phi(\Cd) $.

Let us consider the image of a Hermite function $\ph_k$ under the Bargmann transform. 
  Hermite functions appear as the eigenfunctions $\{\ph_k\}_{k\in\N^d}\subset L^2(\Rd)$ of 
 the harmonic oscillator
\[
- \tfrac{\hbar^2}2\Delta_q + \tfrac12|q|^2, \quad q \in \Rd,
\]
and one can show that
\[
\B\ph_k(q+ip) = \frac{1}{(\pi\hbar)^{d/2} \sqrt{2^{|k|+d}k!} }  \( \frac{z}{\sqrt{\hbar}} \)^k
\]
is an analytic monomial, e.g. by invoking the more general formula in~\cite[Proposition 5]{LT14}. In particular, $\B\ph_k$ is normalized and 
\begin{equation}\label{eq:mon_basis}
\{\B\ph_k\}_{k\in\N^d} \subset \F(\Cd)
\end{equation}
is an orthonormal basis for $\F(\Cd)$ consisting of monomials. This 
property is characteristic to Hermite functions, see e.g. \cite{Ja05}.

 The Fock space $\F(\Cd) $ is a reproducing kernel Hilbert space, and the reproducing kernel
 can be  computed explicitly via the Hermite monomial basis~\eqref{eq:mon_basis} as
 \begin{equation}\label{eq:reprod_kernel}
 \rho(z,w) = \sum_{k\in \N^d}  \overline{\B\ph_k(z)} {\B\ph_k(w)}= (2 \pi \hbar)^{-d} \e^{\overline z w /2\hbar}.
 \end{equation}
That is, for all $z\in\C^d$ and $F\in \F(\Cd)$ one has the pointwise evaluation property
 \begin{equation}\label{eq:reprod_kernel_eval}
F(z) = \lw F(\circ), \rho(z,\circ) \rw_{L^2_\Phi(\Cd)}
\end{equation}
and, as a consequence, one obtains the  derivative formula
 \begin{equation}\label{eq:derivative_kernel}
\frac{\dd^k}{\dd z^k}F(z) = (2\pi\hbar)^{-d} (2\hbar)^{-|k|} \langle F(\circ), \circ^{k}  \rho(z,\circ)  \rangle_{L^2_\Phi(\Cd)}
\end{equation}
 for all $k\in \N^d$.

The Bargmann transform $\B$ can be seen as the complex equivalent of a specific short-time Fourier transform, which for a general
window function  $u\in \mathcal{S}(\Rd)$ is defined as
\begin{equation}\label{eq:stft}
V_{u} \psi (q,p) =  (2\pi\hbar)^{-d/2} \lw \psi,   M_p T_q  u \rw_{L^2(\Rd)}
\end{equation}
with $(q,p)\in\Rdd$ and the standard translation and modulation operators 
\[
T_q\psi(x) = \psi(x-q), \quad M_p \psi(x) =\e^{\i px/\hbar}\psi(x)  , \quad \psi \in L^2(\Rd).
\]
Namely, for the case of a Gaussian window $g_0 := \ph_0$ centered in the origin one observes
 \begin{equation}\label{eq:STFT-Bargmann}
 V_{g_0}\psi(q,-p) = \e^{(\i qp - |z|^2/2)/2\hbar } \B\psi(q+\i p).
 \end{equation}

\subsection{Toeplitz, Weyl and anti-Wick operators}
Let us recall the definitions and basic properties of three quantization schemes: 
Toeplitz, Weyl, and anti-Wick quantization.

The \emph{Toeplitz operator} $\T(m)$ with symbol $m:\C^d \to \C$ is defined by multiplication with $m$ and 
subsequent projection down to the Fock space $\F(\Cd)$:
\begin{equation}\label{eq:toepl}
\T(m) = \P m \P 
\end{equation}
or, more explicitly,
\[
\T(m)F(w) = (2\pi\hbar)^{-d} \int_\Cd m(z) F(z) \e^{\overline z w /2\hbar} \e^{-|z|^2/2\hbar} \dd z
\]
for any $F \in \F(\Cd)$. For $m\in L^\infty(\C^d)$, the quantized operator $\T(m)$ is bounded
on the Fock space $\F(\Cd)$. For more general mapping results we refer to \S\ref{sec:Fock_and_modulation}.

Weyl quantization or canonical quantization appears as the natural 
quantization scheme connecting classical and quantum mechanics. Here,
a function $a:\Rdd~\to~\C$ is associated with the \emph{Weyl quantized operator}
$\op(a)$ via
\begin{equation}\label{eq:weyl_quant}
(\op(a)\psi)(q) = (2\pi \hbar)^{-d} \int_\Rdd a(\tfrac12(y+q),p) \e^{\i (q-y)p/\hbar} \psi(y) \dd y\, \dd p
\end{equation}
where $\Rdd\cong T^*\Rd$ is the phase space of classical mechanics.
The associated phase space representation of quantum states (or signals) is provided by
cross-Wigner functions
\begin{equation}\label{eq:wigner}
\W(\psi,\phi)(q,p) = (2\pi\hbar)^{-d} \int_\Rd \e^{\i p y/\hbar} \psi(q-\tfrac{y}2)
\overline\phi(q+\tfrac{y}2)\, \dd y, \quad (q,p)\in \Rdd.
\end{equation} 
That is, for suitable $a$, $\psi$ and $\phi$, one has
\begin{equation}\label{eq:wigner_weyl}
\lw \op(a)\psi,\phi \rw_{L^2(\Rd)} = \int_\Rdd a(z) \W(\psi,\phi)(z) \dd z,
\end{equation}
where we choose the inner product to be left-linear.
We note that $\W(\psi,\phi)\in~L^2(\Rdd)$ whenever $\psi,\phi \in L^2(\Rd)$. 
In the case $\psi = \phi$ we write
\[
\W(\psi,\psi) =: \W_\psi
\]
for the Wigner function to abbreviate notation.

Despite of their many remarkable properties, Wigner functions $\W_\psi$  
exhibit the drawback of attaining negative values whenever $\psi$ is not a Gaussian, see~\cite{SC83,J97},
and hence typically are not probability densities. However, one can turn  $\W_\psi$ into a nonnegative function by convolution with
another Wigner function: For all $\psi \in L^2(\Rd)$ and Schwartz class windows $\phi \in \mathcal{S}(\Rd)$ with 
$\|\psi\|_{L^2(\Rd)} = \|\phi\|_{L^2(\Rd)} = 1$ the convolution
\[
S^\phi_\psi := \W_\psi * \W_\phi: \Rdd \to \R
\]
is a smooth probability density on phase space,  as can be deduced from \cite[Proposition 1.42]{F89}.
In time-frequency analysis $S^\phi_\psi$ is called a \emph{spectrogram} of $\psi$;
see, e.g., the introduction in~\cite{F13}. Spectrograms constitute a subset of Cohen's class
of phase space distributions; see \cite[\S3.2.1]{F99}.

A popular window function for spectrograms is provided by the
Gaussian wave packet or coherent state
\begin{equation}\label{eq:gaussian_wp}
g_{(q,p)}(x) = (\pi\eps)^{-d/4} \exp\(- \tfrac{1}{2\eps}|x-q|^2 + \tfrac\i\eps p\cdot (x-\tfrac12q)\), \quad (q,p)\in \Rdd, 
\end{equation} 
centered in $(q,p)$. We denote the Gaussian wave packet centered in the origin $(0,0)$ by $g_0$.
The corresponding spectrogram
\begin{equation}\label{eq:husimi_def}
S_\psi^{g_0}(z) = \int_\Rdd \W_\psi(w) (\pi\eps)^{-d} \e^{-|z-w|^2/\eps}~\dd w
\end{equation}
 is known as the \emph{Husimi function} of $\psi$, first introduced in~\cite{H40}. 
Note that
\begin{equation}\label{eq:husimi_exp}
\int_\Rdd a(z) S_\psi^{g_0}(z)  \dd z =\int_\Rdd (\W_{g_0}*a)(z) \W_\psi(z) \dd z =   \lw \opaw (a) \psi,  \psi\rw,
\end{equation}
where $\opaw(a) = \op(\W_{g_0}*a)$ is the so-called \emph{anti-Wick quantized} 
operator associated with $a$; see \cite[\S2.7]{F89}.
From~\cite[Proposition 5]{LT14} we know that the Husimi functions
of  Hermite functions $\{\ph_k\}_{k\in\N^d}$ are given by the formula
\begin{equation*}\label{eq:spec_herm_hus}
S^{g_0}_{\ph_k}(z) = S^{\ph_k}_{g_0}(z) = (2\pi \eps)^{-d} \frac{\e^{-|z|^2/2\eps}}{(2 \eps)^{|k|} k!} |z|^{2k}. 
\end{equation*}

In time-frequency analysis, general anti-Wick type operators
 $\opaw^\ph(a)$, (usually) with  a Schwartz class window  $\ph$, are 
known as \emph{localization operators}. Here, they are equivalenty defined via
multiplication in the image space of the corresponding short-time Fourier transform~\eqref{eq:stft},
\begin{equation}\label{eq:general_awquant}
\opaw^\ph(a) = V_\ph^*aV_\ph, \quad \opaw^{g_0}(a) = \opaw(a),
\end{equation}
where  $a$ denotes both the symbol and the multiplication operator.
The non-negative phase space density corresponding to this quantization scheme
then in turn is given by the spectrogram $S^{\ph}_\psi$, see~\cite{BCG04}.

\subsection{The spectrogram expansion}\label{sec:spec_exp}
In the past decades there has been considerable research  on
the connection between different quantization schemes and their respective
calculi, such as
the classic comparisons of left, right and Weyl quantization
as well as anti-Wick operators, see e.g. \cite[\S2.3 and \S2.4]{L10} 
or \cite[\S4 and \S13]{Z12} for summaries.

Explicit formulas for the Wigner and Husimi functions of general wave packets have been
derived in~\cite{LT14} and subsequently applied in~\cite{KLO15} in order to derive
second order corrections in the comparison of Wigner and Husimi functions.
In~\cite{K17} these corrections have been generalized to arbitrary order by 
proving the following spectrogram expansion.

\begin{thm}[Spectrogram expansion from \cite{K17}]\label{thm:spec_exp}
Let $\psi \in L^2(\Rd)$, $N\in \N$, and $\hbar>0$. 
Then, if one defines the following real-valued phase space function $\mu_\psi^{N}$
in terms of Hermite spectrograms,
\begin{equation}\label{eq:def_mudens}
\mu_\psi^{N}(z) = \sum_{j=0}^{N-1} (-1)^j C_{N-1,j} \sum_{\substack{k\in \N^d \\ |k|=j}} S_\psi^{\ph_k}(z), \quad
C_{k,j} = \sum_{m=j}^k 2^{-m} {d-1+m \choose d-1 + j},
\end{equation}
for any Schwartz function $a:\Rdd \to \C$ there is a constant $C\geq 0$ such that
\begin{equation}\label{eq:spec_approx}
\bigg| \int_\Rdd a(z) \W_\psi(z)\dd z- \int_\Rdd a(z) \mu_\psi^{N}(z) \dd z \bigg| \leq C \hbar^{N} \|\psi\|^2_{L^2(\Rd)} ,
\end{equation}
where $C$ only depends on bounds on  derivatives of $a$ of degree $2N$ and higher.
In particular, if $a$ is a polynomial with $\text{deg}(a)<2N$ then~\eqref{eq:spec_approx} vanishes.
\end{thm}

Retracing the proof for Theorem~\ref{thm:spec_exp} in \cite{K17} 
immediately shows that the 
offdiagonal version of the above approximation holds as well. That is,
\begin{eqnarray}\label{eq:approx_offdiag}
\lw \op(a) \psi,\phi \rw_{L^2} &=&\nonumber \int a(z) \W(\psi,\phi)(z)\dd z \\
&=&  \int_\Rdd a(z) \mu^{N}(\psi,\phi)(z) \dd z + O(\hbar^N)
\end{eqnarray}
with the offdiagonal phase space representation
\begin{equation}\label{eq:offdiag_densitz_approx}
\mu^{N}_{\psi,\phi}(z) = \sum_{j=0}^{N-1} (-1)^j C_{N-1,j} \sum_{\substack{k\in \N^d \\ |k|=j}} \W_{\ph_k}* \W(\psi,\phi)(z)
\end{equation}
of any two functions $\psi,\phi\in L^2(\Rd)$.
We note, however, that $\mu^{N}(\psi,\phi)$ typically has a non-constant 
complex phase and, in particular, is
not a finite linear combination of probability densities.

 In \S\ref{sec:weyl_toeplitz}, polyanalytic Toeplitz operators are applied to prove a statement equivalent to
Theorem~\ref{thm:spec_exp} in polyanalytic Bargmann-Fock spaces.
This yields a variety of new connections between real and complex Weyl, anti-Wick and Toeplitz type
quantization schemes.

\section{Polyanalytic Toeplitz operators}\label{sec:poly_Toeplitz}

In this section, we first recall the definition of polyanalytic Fock-Bargmann spaces and
subsequently introduce and investigate polyanalytic Toeplitz operators which naturally act on these spaces.

\subsection{Polyanalytic Bargmann-Fock spaces}\label{sec:poly_bergmann}

Recall that every polyanalytic function $F$ of order $k\in\N^d$ can be uniquely written as
\[
F(z) = \sum_{\ell \leq k} \bar z^\ell f_\ell(z)
\]
where $f_\ell$, $\ell \in \N^d$, are analytic functions and the sum runs over all multiindices with $0\leq \ell_1\leq k_1,\hdots,0\leq \ell_d\leq k_d$. For all
$k\in \N^d$ we denote
by 
\[
\mathfrak{F}^k(\C^d) = \Big\{ F:\C^d\to\C: F \text{ polyanalytic of degree } k\in\N^d \text{ and } \| F\|_{L^2_\Phi}  < \infty\Big\} 
\]
 the polyanalytic Bargmann-Fock space of degree $k$ which, as we will detail later, has an orthogonal decomposition into true polyanalytic Bargmann-Fock spaces. Note, that
 polyanalytic functions satisfy a generalized Cauchy-Riemann equation of the form
 \[
 \partial_{\overline z_1}^{k_1+1}\cdots  \partial_{\overline z_d}^{k_d+1} F(z) = 0 \Longleftrightarrow  F:\Cd\to\C~\text{is polyanalytic of degree }k.
 \]
 
For later reference let us define  ``translations'' in Bargmann-Fock
spaces by
\begin{equation}\label{eq:fock_translation}
\Theta_z \B f(w) = \B M_p T_q f(w),\quad z = q+ip,
\end{equation}
such that 
\[
\Theta_z F(w) =  (2\pi\hbar)^{-d/2}  \e^{\i pq/2\hbar  - |z|^2/4\hbar +  z  w/2\hbar} F(w- \overline z), \quad w \in \Cd. 
\]

By once again closely following~\cite{Ab10}, we then define polyanalytic 
Bargmann transforms as follows.

\begin{mdef}[Polyanalytic Bargmann transform]
For $k\in \N^d$, the \emph{polyanalytic Bargmann transform} $\B_k: L^2(\Rd) \to \mathfrak{F}^k(\C^d) $ 
of degree $k$ is defined as
\[
\B_kf(z) := \frac{1}{\sqrt k!(2\hbar)^{|k|/2}} \e^{|z|^2/2 \hbar} \frac{\dd^k}{\dd z^k}\Big(\e^{-|z|^2/2 \hbar} \B f(z)  \Big)
\]
in analogy to the definition of Hermite polynomials via their generating function. 
\end{mdef}

As a next step, let us compare $\B^k$ with the short-time
Fourier transform associated with the $k$-th Hermite function as window 
 just as the zero'th order comparison~\eqref{eq:STFT-Bargmann}.

\begin{lem}[see e.g. \cite{Ab10}]\label{lem:STFT_bargmann}
For all $k\in\N^d$ it holds
\[
V_{\ph_k} f(q,-p) = \e^{\i pq/2\hbar - |z|^2/4\hbar } \B_kf(z) 
\]
with $z = q+\i p$. In particular, $\B_k: L^2(\Rd) \to \mathfrak{F}^k(\C^d)$ is a partial isometry.
\end{lem}

\begin{proof}
By utilizing the partial isometry property of the Bargmann transform and recalling the translation 
formula~\eqref{eq:fock_translation}, for $z = q+\i p$ we compute
\begin{eqnarray*}
  &V_{\ph_k} f(q,-p) &= \lw f,   M_{-p} T_q  \ph_k \rw_{L^2(\Rd)} \\
 & &= \lw \B f,   \Theta_{\overline z}  \B \ph_k \rw_{L^2_\Phi(\Cd)} \\
  & &= \frac{(2 \pi\hbar)^{-d}}{(2\hbar)^{|k|/2} \sqrt{k!} }  \e^{\i pq/2\hbar - |z|^2/4\hbar }   
 \lw \B f(w),   \e^{ \overline z  w} \(w-z\)^k  \rw_{L^2_\Phi} \\
  & &=  \frac{\e^{\i pq/2\hbar - |z|^2/4\hbar }}{(2 \pi\hbar)^{d}(2\hbar)^{|k|/2} \sqrt{k!} }     
 \sum_{0\leq \ell \leq k} {k \choose \ell} (-\overline z)^{k-\ell}
 \lw \B f(w),   \e^{ \overline z  w/2\hbar} w^\ell  \rw_{L^2_\Phi} = (\star)
 \end{eqnarray*}
 which by means of the differentiation formula~\eqref{eq:derivative_kernel} leads to the desired result
 \begin{eqnarray*}
(\star) &=& \frac{(2\hbar)^{|k|/2}}{\sqrt k!} \e^{\i pq/2\hbar - |z|^2/4\hbar }  \sum_{0\leq \ell \leq k} {k \choose \ell} (-\overline z)^{k-\ell} \B f^{(\ell)}(z)\\
 &=& \frac{1}{\sqrt k!(2\hbar)^{|k|/2}} \e^{\i pq/2\hbar - |z|^2/4\hbar } \e^{|z|^2/2\hbar} \frac{d^k}{dz^k}\Big( \e^{-|z|^2/2\hbar} \B f(z) \Big)\\
 &=& \e^{\i pq/2\hbar - |z|^2/4\hbar } \B_kf(z) 
\end{eqnarray*}
with standard multiindex notation. Since $\B f^{(\ell)}$ is analytic for all
$\ell \in\N^d$, $\B_kf$ is polyanalytic of degree $k$ and the partial isometry property
of the polyanalytic Bargmann transforms $\B_k$ follows directly from the corresponding
property of the STFT.
\end{proof}

Note that Hermite functions can  be used to construct orthonormal bases 
for polyanalytic function spaces. 
Namely, the set of transformed Hermite functions
\begin{equation}\label{eq:polyanalatic_basis}
\{\B_\ell \ph_m\}_{\ell \leq k, m\in \N^d}, \quad \B_\ell \ph_m(z) \propto z^{m-\ell} \prod_{j=1}^d L_{\ell_j}^{(m_j-\ell_j)}(\tfrac{1}{2\hbar} |z_j|^2) ~~~\text{for}~m\geq \ell,
\end{equation}
and analogously for $m\leq \ell$
is an orthonormal basis of $\mathfrak{F}^k(\C^d) $ for all $k\in \N^d$,
where  $L_n^{(m)}$ denote the Laguerre polynomials, see e.g.~\cite{Ab10}. Formula~\eqref{eq:polyanalatic_basis}
can be proven by using the Laguerre connection for overlap integrals of two shifted Hermite functions similar as for the computation of Wigner transforms of Hermite functions, see e.g. \cite{LT14}. The polynomials in~\eqref{eq:polyanalatic_basis} are particular 
examples of  so-called \emph{special Hermite functions}, see also~\cite{RT09}. 

The polyanalytic Bargmann-Fock spaces admit a decomposition in terms of 
\emph{true polyanalytic Bargmann-Fock spaces}
\[
\F^k(\C^d) := \text{Span}\{\B_k \ph_m\}_{m\in\N^d},\quad \F^0(\C^d) =\F(\C^d),
\]
namely as the orthogonal sum
\[
\mathfrak{F}^k(\C^d) = \bigoplus_{\ell \in\N^d, \ell \leq k} \F^\ell (\C^d).
\]
In particular, recalling~\eqref{eq:polyanalatic_basis} we know that for all $m\in\N^d$ the basis function $\B_\ell \ph_m $ is a polynomial of degree $\ell$ in $\overline z$ which
implies that all nonzero elements of $\F^k(\C^d)$ share this property as well.

The polyanalytic Bargmann transform $\B^k$ acts as an isometric isomorphism
\[
\B_k:L^2(\Rd) \to \F^k(\C^d)
 \] 
and, hence, is also known as~\emph{true} polyanalytic Bargmann transform 
of degree $k$. In analogy to~\eqref{eq:bergman_proj_def}, the map
\[
\P_k := \B_k \B_k^*:L_\Phi(\Cd)\to \F^k(\C^d), \quad \P= \P_0
\]
 is the \emph{polyanalytic  Bergman projector} and its kernel the \emph{polyanalytic Bergman kernel}.
 The reproducing kernel of $\F^k(\C^d)$ is given by
 \[
 \rho^k(z,w) = (2\pi\hbar)^{-d} \prod_{j=1}^d L_{k_j}(\tfrac{1}{2\hbar} |z_j-w_j|^2)\e^{\overline z w/2\hbar}
 \]
 where $L_k$ denotes the $k$th Laguerre polynomial.
 We furthermore use the notion of \emph{polyanalytic Bargmann-Fock spaces of total degree} $n\in\N$,
 \[
\FF^n(\Cd)   := \bigoplus_{|k|= n} \F^k (\Cd),
 \]
 and the corresponding  Bargmann transforms
 \[
 \BB_n:L^2(\Rd)\to  \FF^n(\Cd), \quad \BB_n := \sum_{|k|= n} \B_n
 \]
 on the span of the true polyanalytic functions of total degree $n\in\N$ that is 
 related to
 the Bargmann transform for vector-valued signals from \cite{Ab10} with applications to multiplexing. We denote
 the corresponding \emph{polyanalytic Bergman projector of total degree} $n\in\N$ by
 \[
\PP_n :=  \BB_n \BB_n^* , \quad n\in\N.
 \]
 
 Note that one has the following property:
\begin{lem} 
The polyanalytic Bergman projector of total degree $n\in\N$ satisfies
\[
\PP_n  \neq \sum_{|k|=n} \P_k  
 \]
\end{lem}
This follows since in general for the mixed terms it holds $\B_k \B_\ell^* \neq 0$ though 
 for $k\neq \ell$ one still has the orthogonality property $\B_k^* \B_\ell=0$.

\subsection{Polyanalytic Toeplitz quantization}\label{sec:polyanalytic_toepl}

Recall from~\eqref{eq:general_awquant} that general anti-Wick or localization operators are given by 
\begin{eqnarray*}
\opaw^\ph(a)\psi  &=& \op(\W_{\ph}*a) \psi   \\
&=&  V_{\ph}^*  a V_\ph\psi 
\end{eqnarray*}
where $a$ here  denotes both the phase space function $a$ and the operator of multiplication with $a$. 
Expectation values of anti-Wick operators are computed on the phase space
via the corresponding spectrogram: 
\[
\lw \opaw^\ph(a)  \psi , \psi\rw = \int_\Rdd a(z) S^{\ph}_\psi(z) \dd z.
\]

In the following, we extend the concept of Toeplitz operators as, e.g., defined in~\cite[\S13]{Z12} 
from~\eqref{eq:toepl} to the $d$-dimensional polyanalytic setting, 
see also~\cite{F11} for discussions in the one-dimensional case. 
For defining the quantization, we utilize the polyanalytic Bergman projectors defined in~\S\ref{sec:poly_bergmann}.

\begin{mdef}[Polyanalytic Toeplitz quantization]\label{def:toeplitz_poly}
Let $k\in \N^d$, $n\in \N$ and $f\in L^{\infty}(\Cd)$. Then, the bounded operator
\[
\T_k(f) := \P_k f \P_k, \quad \T_k(f): \F^k(\Cd) \to  \F^k(\Cd)
\]
is called the \emph{true polyanalytic Toeplitz quantization of degree $k$} and
\[
\emph\TT_n(f) := \emph{\PP}_n f \emph\PP_n, \quad \emph\TT_n(f): \emph\FF^n(\Cd)  \to  \emph\FF^n(\Cd) 
\]
\emph{true polyanalytic Toeplitz quantization of total degree $n$}. 
\end{mdef}
For the quantization of more general symbols $f$ one needs to introduce corresponding 
Sobolev type subspaces of  $ \F^k(\Cd)$ with 
stronger decay conditions, as we discuss in \S\ref{sec:Fock_and_modulation}. 

Note that the Bergman projector on the right-hand side of the multiplication operator in
Definition~\ref{def:toeplitz_poly} can be safely ommited when 
acting on polyanalytic Bargmann-Fock spaces. It is included in order
to support the intuition that real-valued symbols $f\in L^{\infty}(\Cd,\R)$ give rise to self-adjoint operators.

For later reference, we also define an off-diagonal type polyanalytic Toeplitz quantization
by multiplication in the polyanalytic space $\F^k(\Cd)$ and projection back to the usual Fock space $\F(\Cd)$.

\begin{mdef}[Projected polyanalytic Toeplitz quantization]
Let $k\in \N^d$ and $f\in L^{\infty}(\Cd)$. Then, the bounded operator
\[
\T_{k,0}(f) := \B \B_k^* f \B_k \B^*, \quad \T_{k,0}(f): \F(\Cd) \to  \F(\Cd)
\]
is called the~$k$-\emph{projected polyanalytic Toeplitz quantization} of $f$.
\end{mdef}


Polyanalytic Toeplitz operators and anti-Wick  quantization are closely related in the following way: Let 
$f \in L^\infty(\C^d)$ and $u,v \in \F^k(\Cd)$, $\B_k^*u=:\phi$ and $\B_k^*v=:\psi$ 
where  $\phi,\psi\in L^2(\Rd)$. Then, one computes
\begin{eqnarray}
\lw u, \P_k f \P_k v \rw_{L^2_\Phi(\Cd)} &=&\nonumber \lw \phi, \B_k^*f\B_k \psi  \rw_{L^2(\Rd)} \\
&=&\label{eq:relation_Toeplitz_aw} (-1)^d \lw \phi, V_{\ph_k}^* \breve f  V_{\ph_k} \psi \rw_{L^2(\Rd)}
\end{eqnarray}
where we define 
\begin{equation}\label{eq:flip_imag_part}
\breve f(q,p) := f(q-ip) .
\end{equation}
For later purposes, let us also define the ``inverse action'' of this map as
\begin{equation}\label{eq:complex_flip_imag_part}
\widehat u(z) := u(q,-p), \quad z=q+ i p \in \Cd, \quad u:\Rdd \to \C.
\end{equation}

Relation~\eqref{eq:relation_Toeplitz_aw} supports the intuition that  localization quantization~\eqref{eq:general_awquant}
with Hermite function windows can be seen as the real-valued equivalent of polyanalytic Toeplitz quantization, 
see also~\cite{F11}.


\section{Polyanalytic Sobolev-Fock spaces and isomorphism theorems}\label{sec:isom}

In this section, we first provide a short overview on Sobolev-Fock and modulation spaces
that serve as a general class of  spaces with natural mapping properties for Toeplitz and localization operators, respectively.
Afterwards, we present the polyanalytic generalizations of those spaces and, as a main result, 
prove an isomorphism theorem for polyanalytic Toeplitz operators.

\subsection{Modulation spaces and Sobolev-Fock spaces}\label{sec:Fock_and_modulation}

Let us briefly review \emph{modulation spaces} and and their images under the 
Bargmann transform, the so-called \emph{Sobolev-Fock spaces}. Modulation spaces form a natural framework
for the calculus of localization operators in the same way as Sobolev-Fock spaces do for Toeplitz operators.

Following usual conventions, see e.g. \cite{GT13}, we call a locally bounded weight function $m:\Rdd \to (0,\infty)$ \emph{moderate} if
\[
\sup_{z\in\Rdd}\( \frac{m(z+y)}{m(z)},\frac{m(z-y)}{m(z)}\)=:v(y)<\infty \quad \text{for all } y\in\Rdd. 
\]
As a result, $v$ is a submultiplicative function and $m$ satisfies
\[
m(z+y)\leq m(z)v(y) \quad \text{for all } z,y\in\Rdd.
\]
We restrict ourselves to  weights of polynomial growth and call a weight function \emph{admissable} if it is moderate, 
continuous and at most of polynomial growth.
For any fixed submultiplicative weight function $v:\Rdd \to (0,\infty)$ we define the set of $v$-admissable weights
as
\[
\mathcal{M}_v = \{ m \in L^\infty_{loc}(\Rdd)~\text{admissable and }  0 < m(z+y) < m(z)v(y) \; \forall \,y,z\in \Rdd\}.
\]
Then, the modulation spaces with admissible weight $m$ are defined as
\[
M_{m}^{p,q}(\Rd)= \Big\{ f\in \mathcal{S}(\Rd)': \( \int_\Rd \( \int_\Rd |V_{g_0}f(x,\xi)|^p m(x,\xi)^p \dd x \)^{q/p} \dd \xi \)^{1/q} < \infty\Big\},
\]
$1\leq p,q \leq \infty$, 
and contain functions (or distributions) that show controlled growth properties together with their Fourier transforms. We note that modulation spaces 
do not change if we replace the Gaussian window $g_0$ by a different Schwartz function, see e.g.~\cite[\S 11]{GR01}.

Similarly as the classical Fock space $\F(\Cd)$ is the image of $L^2(\Rd)$ under the Bargmann transform, one can look at Fock-type 
spaces that are the equivalents of modulation spaces in the complex setting.
We use the notation from \cite{GT13} and write $\mathcal{M}_v^\C$ for 
complex $v$-admissable weights with $v:\Cd \to (0,\infty)$ moderate. We introduce for any complex 
moderate weight $m$ the Sobolev-Fock spaces 
\[
\F^{p,q}_m(\Cd) = \Big\{ F:\Cd \to \C \text{ entire and }\|F\|_{L^{q,p}_{\Phi,m}}<\infty  \Big\}
\]
that are complete subspaces  of the Banach spaces $L^{p,q}_{\Phi,m}$ with the weighted mixed $p,q$-norm
\[
\|F\|_{L^{p,q}_{\Phi,m}}  =  \(\int_\Rd  \(\int_\Rd |F(z)|^p  m(z)^p \e^{-p|z|^2/4\hbar } \dd \Re(z)\)^{q/p} \dd \Im(z)\)^{1/q}
\]
consisting of entire functions.
In particular, $\F^{2,2}_1(\Cd) = \F(\Cd)$ gives the usual Fock space. It is well-known, see e.g.~\cite{GT11,GT13}, that the Bargmann transform $\B$
maps the modulation space  $M_{m}^{p,q}(\Rd)$ isometrically to the Sobolev-Fock 
space $\F^{q,p}_{\breve m}(\Cd)$, where we employ the notation from~\eqref{eq:flip_imag_part}.
In particular, 
from~\cite[Theorem 5.4]{GT11}  we are able to rephrase the following result.
\begin{lem}\label{lem:boundedness_toeplitz}
Let $\mu \in \mathcal{M}^\C_w$ and $m\in \mathcal{M}^\C_v$. Then, for all $1\leq q,p \leq \infty$, 
the Toeplitz
operator $\T(m)$ is a bounded, invertible map from $\F^{p,q}_\mu(\Cd)$ to $\F^{p,q}_{\mu/m}(\Cd)$.
\end{lem}

\subsection{Polyanalytic Sobolev-Fock spaces}\label{sec:_poly_sobolev_Fock}

Based on the analytic Sobolev-Fock space theory suitable for Toeplitz operators from~\S\ref{sec:Fock_and_modulation}
one can define similar function spaces in the polyanalytic setting.
For any $k\in\N^d$ we closely follow the definitions in~\cite{AG10} and 
define \emph{true polyanalytic Sobolev-Fock spaces} with mixed $p,q$-norms
as 
\[
\F^{k,p,q}_m(\Cd) = \Big\{ F:\Cd \to \C \text{ true polyanalytic of degree } k \text{ and }\|F\|_{L^{p,q}_{\Phi,m}}<\infty  \Big\}
\]
where $\F^{0,p,q}_m(\Cd) = \F^{p,q}_m(\Cd)$. Moreover, we define
\[
\FF^{n,p,q}_m(\C^d) = \Big\{ F:\C^d\to\C \text{ true polyanalytic of total degree } n\in\N,~  \| F\|_{L^{p,q}_{\Phi,m}}  < \infty\Big\}
\]
with $\FF^{n,2,2}_1(\Cd) = \FF^n(\Cd)$.
As we summarize in the following Lemma~\ref{lem:image_poly_barg}, polyanalytic Fock-Sobolev spaces are precisely the image of the usual modulation spaces under the polyanalytic Bargmann transform.

\begin{lem}\label{lem:image_poly_barg}
For all $1\leq p,q \leq \infty$, $k\in\N^d$ and $m\in \M_v$, the polyanalytic Bargmann transform $\B_k$ is 
an isomorphism
\[
\B_k : M_{m}^{p,q}(\Rd) \to \F^{k,q,p}_{\breve m}(\Cd).
\]
\end{lem}

\begin{proof}
For $k=0$ this result is well-known, see e.g. \cite{GT11,GT13,GR01}. For $k\neq 0$ the results follow from 
Lemma~\ref{lem:STFT_bargmann} by observing that the modulation space $M_{m}^{p,q}(\Rd)$ 
can  be defined without harm with the Hermite window $\ph_k$ instead of $g$.
\end{proof}

\begin{rem}
We note that --- as we stick to weight functions of polynomial growth --- the Schwartz space is contained in all
considered modulation spaces. This in particular implies that the span of special Hermite functions
\[
\Span\{\B_\ell \ph_m\}_{|\ell|= n, m\in \N^d} = \Span \Big\{ z^{m-\ell} \prod_{j=1}^d L_{\ell_j}^{(m_j-\ell_j)}(\tfrac{1}{2\hbar} |z_j|^2)\Big\}_{}
\] 
is a dense subset of $\emph\FF^{n,p,q}_m(\C^d)$, see also~\cite{RT09}. Moreover, the basis functions $\B_\ell \ph_m$ are orthogonal if
$m$ is radial in each component, that is, $m(z_1,\hdots,z_d) = \tilde m(|z_1|,\hdots,|z_d|)$ for some $\tilde m$, see also~\cite{GT13}.
\end{rem}

\subsection{Isomorphism results}\label{sec:_poly_isomorphism}

In the following, we generalize the isomorphism result 
from Lemma~\ref{lem:boundedness_toeplitz} to the polyanalytic context.  For this purpose,
we investigate the mapping properties of polyanalytic Toeplitz operators on their respective Sobolev-Fock spaces.
This constitutes a main result of this paper.
%

\begin{thm}
Let $1\leq p,q\leq \infty$, $k\in\N^d$, $\mu \in \M^\C_w$ and $m\in \M^\C_v$ be continuous. Then the polyanalytic 
Toeplitz operator $\T_k(m)$ constitutes an isomorphism as a map
\[
\T_k(m) : \F^{k,p,q}_\mu(\Cd) \to \F^{k,p,q}_{\mu/m}(\Cd)
\]
and the $k$-projected polyanalytic Toeplitz operator $\T_{k,0}(m)$ is an
isomorphism 
\[
\T_{k,0}(m) : \F^{p,q}_\mu(\Cd) \to \F^{p,q}_{\mu/m}(\Cd).
\]
\end{thm}

\begin{proof}
By Lemma~\ref{lem:image_poly_barg} the polyanalytic Bargmann transform
$\B_k$ is an isomorphism as a map
\[
\B_k : M_{\breve m}^{p,q}(\Rd) \to \F^{k,q,p}_{ m}(\Cd).
\]
and the isomorphism property for the localization operator $\opaw^{\ph_k}(\breve f) $ with Hermite function window as a map 
\[
\opaw^{\ph_k}(\breve m) :M_{\breve \mu}^{p,q}(\Rd) \to M_{\breve \mu/ \breve m}^{p,q}(\Rd) 
\]
is well-known, see e.g. \cite[Theorem 4.3]{GT13}. Moreover, from~\eqref{eq:relation_Toeplitz_aw}
we infer that
\[
\B_k^* \T_k(m) \B_k  = (-1)^d  \opaw^{\ph_k}(\breve m).
\]
Hence, $\T_k(m)$ can be written as composition of three isomorphisms
\[
\begin{tikzcd}[row sep = 6mm, column sep = 30mm]
\F^{k,p,q}_\mu \arrow{r}{\T_k(m)}& \F^{k,p,q}_{\mu/m} \\
M_{\breve \mu}^{p,q} \arrow{r}{(-1)^d\opaw^{\ph_k}(\breve m)}  \arrow[swap]{u}{\B_k} &  M_{\breve \mu/ \breve m}^{p,q} \arrow{u}{\B_k}
\end{tikzcd}
\]
which completes the proof for the first part of the assertion. 
For the second part one similarly obtains the diagram
\[
\begin{tikzcd}[row sep = 6mm, column sep = 30mm]
\F^{p,q}_\mu \arrow{r}{\T_{k,0}(m)}& \F^{p,q}_{\mu/m} \\
M_{\breve \mu}^{p,q} \arrow{r}{(-1)^d\opaw^{\ph_k}(\breve m)}  \arrow[swap]{u}{\B} &  
M_{\breve \mu/ \breve m}^{p,q} \arrow{u}{\B}
\end{tikzcd}
\]
for showing the isomorphism property.
\end{proof}

\section{Symbol calculus}\label{sec:poly_toeplitz_calculus}

After we presented the basic concept of polyanalytic Toeplitz operators and
their natural action on polyanalytic Sobolev-Fock spaces in the previous sections,
we now turn towards a basic symbolic operator calculus for Hermite localization operators as well as
polyanalytic Toeplitz operators by providing expansions for compositions and commutators.

For localization operators with symbols in modulation spaces, 
composition formulas and  Fredholm properties have been 
derived in great generality in~\cite{CG06}. Our aim is to obtain more explicit expressions and expansions for small $\hbar$.
We start by presenting asymptotic expansions of 
localization operators with Hermite windows and their compositions as $\hbar \to 0$, before
moving on to polyanalytic spaces and operators.

\subsection{Weyl expansion of Hermite localization operators}\label{sec:herm_locali_comp}

By observing that localization operators are in fact smoothed Weyl operators,
\[
\opaw^\ph(a) = \op(\W_\ph * a)
\]
one can Taylor expand the convolution and 
use the Moyal product expansion in order to derive asymptotic expansions of 
compositions of localization operators.

For the standard case of a Gaussian window we recall the following result from~\cite[Lemma 1]{KL13}.
\begin{lem}\label{lem:expansion_anti-Wick}
Let $a:\Rdd \to \C$ be a Schwartz function, $\hbar>0$, $N\in \N$. Then,
\[
\opaw(a) = \op\(a + \sum_{k=1}^{N-1} \frac{(\hbar \Delta)^k}{4^k k!}a\) 
+ \hbar^N\op(r_\hbar)
\]
with  a  family $r_\hbar$ of Schwartz functions satisfying $\sup_{\hbar>0}\| \op(r_\hbar)\|_{L^2\to L^2}<\infty$.
\end{lem}

Let us generalize this formula for higher order Hermite functions. We do this
by similar means as applied in~\cite{K17} for deriving the expansion with 
Hermite spectrograms. For this purpose, let us recall the formula for
Wigner transforms of Hermite functions,
\begin{equation}\label{eq:wigner_hermite}
\W_{\ph_k} (z) = (\pi\hbar)^{-d} \e^{-|z|^2/\hbar} (-1)^{|k|} 
\prod_{j=1}^d L_{k_j}(\tfrac{2}{\hbar}|z_j|^2),
\end{equation}
where $z=(q,p)\in\Rdd$, $z_j = (q_j,p_j) \in \R^2$ and  
\[
L_k(x) = \sum_{j=0}^k {k \choose k-j} \frac{(-x)^j}{j!}, \quad x\in\R, \,\,k\in \N,
\]
 is the $k$th Laguerre polynomial,
see, e.g.,~\cite[\S 1.9]{F89}. In order to generalize Lemma~\ref{lem:expansion_anti-Wick} to arbitrary Hermite function windows 
we have to first get a better understanding of higher order moments of the Wigner transforms of Hermite functions. Note, that
due to the symmetry of $\W_{\ph_k}$ only
even moments are different from zero.
\begin{prop}\label{prop:2d_hermite_moments}
Let $\a,\beta, k \in\N$ be arbitrary.Then,
\[
\int_{\R^2} x^{2\a} \xi^{2\beta} \W_{\ph_k}(x,\xi) \dd x \dd\xi =  { } _2{\text F}_1(\a+\beta+1,-k;1;2)(-1)^{k} 
  \frac{\hbar^{\a+\beta} (2\a)!(2\beta)!}{4^{\a+\beta}\a!\beta!}
\]
where $_2{\text F}_1$ is the hypergeometric function.
\end{prop}

For the proof of Proposition~\ref{prop:2d_hermite_moments}, which is mainly built on relations of binomial sums and Gamma functions, we refer to Appendix~\ref{app_moments_wigner_hermite}.
Now, we are ready to generalize Lemma~\ref{lem:expansion_anti-Wick} as follows.

\begin{lem}\label{lem:expansion_hermite_localiza}
Let $k,N\in\N^d$, $\hbar>0$ and $a$ being a Schwartz function. Then,
\[
\opaw^{\ph_k}(a) =  \op\(  \sum_{m=0}^{N-1} \frac{\hbar^m D_m}{4^m m!}a\) 
+ \hbar^N\op(r^k_\hbar)
\]
with  a  family $r^k_\hbar$ of functions satisfying $\sup_{\hbar>0}\| \op(r^k_\hbar)\|_{L^2\to L^2}<\infty$
and the order $2m$ phase space differential operator $D_m$ given by
\[
D_m a (z) = (-1)^{|k|} m! \sum_{|\a|=m}  c^{(k)}_\a \partial^{2\a} a(z)
\]
that is a  sum of total order $2m$ differential operators with constant coefficients
\[
c^{(k)}_\a = \frac{1}{\a!} \prod_{j=1}^d { }_2{\text F}_1(\a_j+\a_{j+d}+1,-k_j;1;2).
\]
\end{lem}

\begin{proof}
We can basically retrace the idea of~\cite[Lemma 1]{KL13} by writing 
\begin{equation}\label{eq:conv_a_hermite_wigner}
a * \W_{\ph_k} (z) = \int_\Rdd a(\zeta) \W_{\ph_k} (z - \zeta) \dd \zeta
\end{equation}
 and using a Taylor expansion of $a$ around $z$,
\begin{eqnarray*}
a(\zeta) &=& \sum_{|\alpha|=0}^{2N-1} \frac{(\zeta-z)^\alpha}{\alpha!}(\partial^\alpha a)(z) \\ 
& &+ 
2N \sum_{|\alpha|=2N} \frac{(\zeta - z)^\alpha}{\alpha!} \int_0^1 (1-\theta)^{2N-1} (\partial^\alpha a)(z + \theta(\zeta - z)) \dd \theta.
\end{eqnarray*}
Since the symmetry of~\eqref{eq:wigner_hermite} implies that
\[
\int_\Rdd f(z) \W_{\ph_k} (z) \dd z = 0
\]
whenever $f$ is an odd function, the derivatives of odd degree in the Taylor
expansion of $a$ do not contribute to the integral~\eqref{eq:conv_a_hermite_wigner}.
For the even degree polynomials we apply Proposition~\ref{prop:2d_hermite_moments} and compute 
\begin{align*}
\sum_{|\alpha|=m}&  \int_\Rdd \frac{(\zeta-z)^{2\alpha}}{2\alpha!} 
(\partial^{2\alpha} a)(z) \W_{\ph_k} (z-\zeta) \dd \zeta \\
&= \sum_{|\alpha|=m} \frac{(\partial^{2\alpha} a)(z)}{2\alpha!} \int_\Rdd \zeta^{2\alpha} \W_{\ph_k} (\zeta) \dd \zeta \\
&= (-1)^{k} \frac{\hbar^m}{4^m}   \sum_{|\alpha|=m} \frac{(\partial^{2\alpha} a)(z)}{\alpha!} \prod_{j=1}^d { }_2{\text F}_1(\a_j+\a_{j+d}+1,-k_j;1;2)
\end{align*}
by utilizing the fact that the Wigner function factorizes in the form~\eqref{eq:wigner_hermite}. Hence,
\begin{align}
a * \W_{\ph_k} (z) &= \sum_{m=0}^{N-1} \frac{\hbar^m D_m}{4^m m!}a + \hbar^N r^k_\hbar
\end{align}
which completes the proof as the Calderon-Vaillancourt theorem implies the uniform boundedness of $r^k_\hbar$.
\end{proof}

We would like to stress that due to the fact that the coefficients 
$c_\alpha$ are varying in $\alpha$ it is not straight-forward to write down 
an inverse expansion as in general $D_m D_n \not\propto D_{m+n}$ unless in
the Husimi case $k=0$.

\subsection{Compositions and commutators of Hermite localization operators}\label{sec:comp_localization}

Recall that the composition of two Weyl quantized
operators is a Weyl quantized operator again, with the symbol given by the
famous~\emph{Moyal product} $\sharp$ of the two symbols,
\begin{equation}\label{eq:moyal_prod}
\op(a)\op(b) = \op(a\sharp b),
\end{equation}
see, e.g.,~\cite[\S4.3]{Z12}. In contrast, the product of two localization operators
typically is  not a localization operator again. 
However, the product can be expanded as a sum of localization operators with
a regularizing operator as error term that becomes arbitrary small as $\hbar\to 0$, see~\cite{CG06}.

Based on the expansion from Lemma~\ref{lem:expansion_hermite_localiza},
we obtain the following Weyl composition formula for two localization operators 
that employs the operator $A(\nabla)$,
\[
A(\nabla)f(z,w) = \tfrac12 \sigma(\nabla_w,\nabla_z) f(z,w), \quad z,w\in \Rdd,
\]
acting on the doubled phase space, where $\sigma$ denotes the standard symplectic form.
Note that $A(\nabla)$ is the generator of bidifferential operators 
\begin{equation}\label{eq:bidiff_moyal}
\alpha_n(a,b) = \Big[ A(\nabla)^n(a\otimes b) \Big]_{\text{diag}}, \quad n\in \N,
\end{equation}
that define the Moyal product expansion.

\begin{prop}[Composition of localization operators]\label{prop:localization_comp_general}
Let $k,N\in\N^d$, $\hbar>0$ and $a$ and $b$ be Schwartz functions. Then,
\[
\opaw^{\ph_k}(a)\opaw^{\ph_k}(b) =  \op\( \sum_{j=0}^{N-1} \frac{\hbar^j}{4^j} \(\sum_{n+m+\ell=j}  \frac{C_{n,m,\ell} (a,b)}{m!n!\ell! } \) \)
+ \hbar^N\op(\rho^k_\hbar)
\]
with  a  family $\rho^k_\hbar$ of Schwartz functions satisfying $\sup_{\hbar>0}\| \op(\rho^k_\hbar)\|_{L^2\to L^2}<\infty$
and the total order $2(n+m+\ell)$ bidifferential operators
\[
C_{n,m,\ell} (a,b) =   [(-4\i )^\ell\alpha_\ell (D_ma, D_nb) ]_{\text{diag}}
\]
where $D_n$ has been defined in Lemma~\ref{lem:expansion_hermite_localiza} and $\alpha_\ell$ in \eqref{eq:bidiff_moyal}.
\end{prop}
\begin{proof}
We apply Lemma~\ref{lem:expansion_hermite_localiza} and the expansion of the Moyal product $\sharp$ to compute
\begin{eqnarray*}
& &\opaw^{\ph_k}(a)\opaw^{\ph_k}(b)  = \op\Big(\sum_{m+n=0}^{N-1}   \frac{\hbar^{n+m} }{4^{m+n}m!n! } D_ma \sharp D_n b\Big) + \hbar^N\op( \varrho^k_\hbar) \\
& &\quad= \sum_{j=0}^{N-1} \frac{\hbar^j}{4^j} \op\(\sum_{n+m+\ell=j}  \frac{(-4\i )^\ell}{m!n!\ell! }   [A(\nabla)^\ell (D_ma \otimes D_nb) ]_{\text{diag}} \) + \hbar^N\op(\rho^k_\hbar)
\end{eqnarray*}
where $\varrho^k_\hbar$ and  $\rho^k_\hbar$  are families of Schwartz functions giving rise to  uniformly bounded operator families.
\end{proof}

For illustration purposes, let us look at the general expansion from Proposition~\ref{prop:localization_comp_general}
in the case of second order errors. We compute
\[
D_1a(z) = \sum_{j=1}^{d} (2k_j+1)(\partial^2_j a(z) + \partial^2_{j+d} a(z))
\]
and observe that $D_1$ is a diagonally weighted Laplace operator on $\Rdd$,
\[
D_1a(z) = \lw \nabla_z, \text{diag}(2k+1,2k+1) \nabla_z\rw a(z) =: \Delta_{(k)} a(z),
\]
where $(2k+1,2k+1) := (2k_1+1,2k_1+1,...,2k_d+1) \in \Rdd$.
Hence, we obtain the following second order composition formula for
two localization operators in terms of a Weyl operator:

\begin{lem}\label{comp_2aw_weyl}
Let $k\in\N^d$, $\hbar>0$ and $a$ and $b$ be Schwartz functions. Then,
\[
\opaw^{\ph_k}(a)\opaw^{\ph_k}(b) =  
\op\( ab + \tfrac{\hbar}{2} \( -\i \sigma(\nabla a,\nabla b) + \tfrac12 b \Delta_{(k)} a   + \tfrac12 a \Delta_{(k)} b  \) \)
+ O(\hbar^2).
\] 
\end{lem}

In fact, if we allow for second order error terms, the expansion from Lemma~\ref{lem:expansion_hermite_localiza}
can be approximately inverted via
\begin{equation}\label{eq:weyl_as_aw_second_order}
\opaw^{\ph_k}(a - \tfrac\hbar4 \Delta_{(k)} a ) =  \op(a) + O(\hbar^2).
\end{equation}
This observation in turn implies the following composition formula for Hermite window
localization operators.

\begin{thm}\label{thm:composition_localization}
Let $k\in\N^d$, $\hbar>0$ and $a$ and $b$ be Schwartz functions. Then, it holds
\[
\opaw^{\ph_k}(a)\opaw^{\ph_k}(b) =  
\opaw^{\ph_k}\( ab - \tfrac\hbar2\(\i \sigma(\nabla a,\nabla b) + \lw \nabla_{(k)} a, \nabla_{(k)} b \rw \) \)
+ \hbar^2 \op(\theta_\hbar^k).
\] 
with  a  family $\theta^k_\hbar$ of Schwartz functions satisfying $\sup_{\hbar>0}\| \op(\theta^k_\hbar)\|_{L^2\to L^2}<\infty$ and the
weighted gradient 
\[
\nabla_{(k)}= (\sqrt{2k_1+1}\partial_1,\sqrt{2k_1+1}\partial_2,...,\sqrt{2k_d+1}\partial_{2d}).
\]
\end{thm}

\begin{proof}
By combining Proposition~\ref{comp_2aw_weyl} and~\eqref{eq:weyl_as_aw_second_order} we obtain
\begin{eqnarray*}
\opaw^{\ph_k}(a)\opaw^{\ph_k}(b) &=  &
\opaw^{\ph_k}\( ab + \tfrac\hbar2(-\i \sigma(\nabla a,\nabla b) + \tfrac12b \Delta_{(k)} a   + \tfrac12a \Delta_{(k)} b - \tfrac12\Delta_{(k)}(ab)) \)   \\
& &+ \hbar^2 \op(\theta_\hbar^k) 
\end{eqnarray*}
where, by the Calderon-Vailloncourt theorem, the  the second order terms in $\hbar$ have a Schwartz class symbol with the desired boundedness properties.
Then, calculating
\[
\Delta_{(k)}(ab) = \Delta_{(k)}a b + a \Delta_{(k)} b + 2 \lw \nabla_{(k)} a, \nabla_{(k)} b \rw
\]
implies the result.
\end{proof}

From Theorem~\ref{thm:composition_localization} we directly infer that the commutator of two localization operators exhibits the same Poisson
bracket property as the Moyal bracket for Weyl operators with the difference that the error is of second instead of third order in $\hbar$.

\begin{cor}\label{cor:commutator_locali}
Let $k\in\N^d$, $\hbar>0$ and $a$ and $b$ be Schwartz functions. Then, for the commutator of Hermite window localization operators
it holds
\[
\lk \opaw^{\ph_k}(a),\opaw^{\ph_k}(b) \rk = \tfrac{\hbar}{\i} \opaw^{\ph_k}\(\lb a,b \rb\) + O(\hbar^2)
\]
where $\lb a,b \rb$ denotes the Poisson bracket.
\end{cor}

\begin{rem}[Hermite star products]
The Hermite star products $\star_k$ can be formally defined as
\[
 \opaw^{\ph_k}(a)\opaw^{\ph_k}(b)  =: \opaw^{\ph_k}(a \star_k b)
 \]
 on the algebra $\mathcal{C}^\infty(\Rdd)[[\hbar]]$ of formal power series in $\hbar$ with smooth coefficients.
Corollary~\ref{cor:commutator_locali} illustrates that --- just as the Moyal product $\sharp$--- 
all Hermite star products $ \star_k$ are compatible with the canonical 
Poisson structure on phase space. Moreover, the expansion from Lemma~\ref{lem:expansion_hermite_localiza}
implies that the differential star products $ \star_k$ and $\sharp$ are equivalent for all $k$ in the sense of deformation 
quantization, see, e.g., \cite{K03,BRW07,Sch18}.
In particular, we note that the bidifferential operator $ \i \alpha_1(a,b) + \tfrac12 \lw \nabla_{(k)} a, \nabla_{(k)} b \rw $ 
from Theorem~\ref{thm:composition_localization} defines the same 2-cocycle as the Moyal bidifferential operator $\i \alpha_1(a,b) $  
in the Hochschild cochain complex over $\mathcal{C}^\infty(\Rdd)[[\hbar]]$ and only differs by the symmetrical coboundary term.
\end{rem}

\begin{rem}[Anti-Wick star product]
In the Husimi case $k=0$ one has $D_n = \Delta^n$ and  can explicitly derive higher order versions
of Theorem~\ref{thm:composition_localization}.
In analogy to the Moyal expansion a simple but tedious calculation yields
\begin{eqnarray*}
\opaw^{\ph_0}(a)\opaw^{\ph_0}(b) =\opaw^{\ph_0}(a \star_0 b) &=  &  \sum_{j=0}^\infty \opaw^{\ph_0}\( \hbar^n \beta_n( a, b) \)
\end{eqnarray*}
with the bidifferential operators
\[
\beta_n( a, b) = \frac{(-1)^n}{2^n n!} \Big[ \(\i \sigma(\nabla_z,\nabla_w) + \lw \nabla_z, \nabla_w \rw\)^n a(z)b(w)\Big]_{\text{diag}},
\]
see~\cite{K12}. In other words, the operator $B(\nabla) = A(\nabla) + \tfrac12 \lw \nabla_z, \nabla_w \rw$ 
generates the bidifferential operators $\beta_n$ that charaterize the anti-Wick star product  $\star_0$.
\end{rem}

The symmetric term $ \lw \nabla_{(k)} a, \nabla_{(k)} b \rw$  creates coboundary terms in the bidifferential operators 
defining Hermite star products and implies that the $O(\hbar^2)$ error
for the commutator expansion in Corollary~\ref{cor:commutator_locali} in general is sharp. In contrast,
for the Moyal
case the antisymmetry of $A(\nabla)$ causes $O(\hbar^3)$ errors for the commutator which 
is the main ingredient for the the well-known Egorov theorem  that allows to link 
quantum and quasi-classical dynamics with $O(\hbar^2)$ errors, see~\cite{BR02,LR10}.

We conclude this section by stressing again that the Weyl operator error term $\op(\theta_\hbar^k)$
in Theorem~\ref{thm:composition_localization} is in general not a localization
operator itself. This makes the composition formula purely asymptotic
in nature.

\subsection{Calculus of polyanalytic Toeplitz operators}\label{sec:polyanalytic_toepl_comp}

The formulas from~\S\ref{sec:comp_localization} also imply composition rules for polyanalytic Toeplitz
operators as they appear as the complex equivalents of corresponding 
localization operators with Hermite function windows.

From~\eqref{eq:relation_Toeplitz_aw} we first recall the translation formulas
\begin{equation}\label{eq:trans_poly_toepl_loc}
 \B_k^* \T_k(f) \B_k  = (-1)^d  \opaw^{\ph_k}(\breve f),\quad \B^* \T_{k,0} (f) \B = (-1)^d   \opaw^{\ph_k}(\breve f) 
 \end{equation}
  between localization operators acting on real-valued signals and polyanalytic Toeplitz operators acting
in the complex domain. Moreover, for convenience we introduce the operators
\[
\Xi_{(k)}(\nabla) f(z,w) = \i \sigma(\nabla_w,\nabla_z) f(z,w) + \lw \nabla_{(k),z}  \nabla_{(k),w} \rw f(z,w), \quad z,w \in \Rdd,
\]
and their complex counterpart
\[
\widehat\Xi_{(k)}(\partial,\overline\partial) F(z,w) = \(  4\i \partial_z \overline \partial_w + \lw \nabla_{\Re z,\Im z},d_k  \nabla_{\Re w,\Im w} \rw_{\Rdd}  \) F(z,w), \quad z,w \in \Cd,
\]
where $\partial_z,\overline\partial_z$ as usual denote complex Wirtinger differentials and the diagonal matrix $d_k = \text{diag}(2k,2k)\in \R^{2d\times 2d}$.
Note that for $k=0$ the second term vanishes and we obtain the simple expression
\[
\widehat\Xi_{(0)}(\partial,\overline\partial) F(z,w) = 4\i \partial_z \overline \partial_w F(z,w).
\]
The operator $\widehat\Xi_{(k)}(\partial,\overline\partial) $ represents $\Xi_{(k)}(\nabla)$ in the complex domain in the following way.
\begin{lem}\label{lem:transition_complex_real_bidiff}
Let $m,\mu:\Cd \to \C$ be smooth and $k\in\N^d$. Then, it holds
\[
\widehat\Xi_{(k)}(\partial,\overline\partial) (m \otimes \mu) (\bar z,\bar w) = \Xi_{(k)}(\nabla) (\breve m \otimes \breve \mu) (z_\R,w_\R),
\]
where $z,w \in \Cd$ and $z_\R = (\Re z, \Im z),w_\R = (\Re w, \Im w)\in \Rdd$.
\end{lem}

\begin{proof}
The proof is a simple calculation that only uses the definition of the complex Wirtinger differentials 
$\partial_z f(z) = \tfrac12(\partial_{\Re z} f(z) - \i \partial_{\Im z}) f(z) $ and $\overline \partial f(z) = \tfrac12( \partial_{\Re z} f(z) + \i \partial_{\Im z})$.
\end{proof}

With this notation in place, we arrive at the following composition and commutator formulas for polyanalytic Toeplitz
operators.

\begin{thm}\label{thm:composition_toeplitz}
Let $k\in\N^d$, $\hbar>0$ and $m,\mu:\Cd \to \C$ be Schwartz class functions. Then, it holds
\[
\T_k(m)\T_k(\mu) =  \T_k(m \mu -\tfrac\hbar 2 [\widehat\Xi_{(k)}(\nabla)( m \otimes  \mu)]_{\text{diag}})
+ \hbar^2 \B_k^*  \op(\theta_\hbar^k) \B_k.
\] 
with  a  family $\theta^k_\hbar$ of Schwartz functions satisfying $\sup_{\hbar>0}\| \op(\theta^k_\hbar)\|_{L^2\to L^2}<\infty$.
\end{thm}

\begin{proof}
We calculate
\begin{eqnarray*}
\T_k(m)\T_k(\mu) &=& \B_k^* \opaw^{\ph_k}(\breve m)\opaw^{\ph_k}(\breve \mu) \B_k \\
&=& \B_k^* \opaw^{\ph_k}(\breve m \breve \mu -\tfrac\hbar2 [\Xi_{(k)}(\nabla)(\breve m \otimes \breve \mu)]_{\text{diag}}) \B_k+ O(\hbar^2)
\end{eqnarray*}
by using~\eqref{eq:trans_poly_toepl_loc} and applying Theorem~\ref{thm:composition_localization}. 
From Lemma~\ref{lem:transition_complex_real_bidiff} we then obtain
\[
[\Xi_{(k)}(\nabla)(\breve m \otimes \breve \mu)]_{\text{diag}}(q,p) = [\widehat\Xi_{(k)}(\nabla)( m \otimes  \mu)]_{\text{diag}}(q-ip)
\]
which completes the proof.
\end{proof}

Let us remark here, that in the usual Toeplitz quantization case $k=0$ this composition formula beautifully reduces to
\[
\T(m)\T(\mu) =  \T(m \mu - 2 \i \hbar \,\partial m \,\overline \partial \mu  )  + O(\hbar^2 ).
\]
Ignoring all growth restrictions, this shows that whenever $\mu$ is analytic (or $m$ antianalytic) the product $m \mu$ is the 
appropriate Toeplitz symbol of the composition up to second order errors.

\section{Weyl quantization and polyanalytic Toeplitz operators}\label{sec:weyl_toeplitz}

Let us revisit the spectrogram expansion from Theorem~\ref{thm:spec_exp} with the objects we have defined and investigated
so far. By recalling the phase space integral formulas~\eqref{eq:wigner_weyl} 
and~\eqref{eq:husimi_exp} we first observe that Theorem~\ref{thm:spec_exp} in fact 
can be read as a weak approximation of Weyl operators in terms of localization operators
with Hermite function windows. 
In this section, we derive  expansions of complex Weyl operators in terms of
Bargmann quantized operators and, thus, prove a complex version of Theorem~\ref{thm:spec_exp}.

\subsection{An anti-Wick expansion of Weyl operators}
By employing the off-diagonal convolution formula
\begin{equation}\label{eq:offdiag_elements_aw}
\W_{u}* \W(\psi,\phi) = V_{u}\psi \; \overline{V_{u}\phi}.
\end{equation}
and  the definition~\eqref{eq:general_awquant} of localization operators
we can rewrite~\eqref{eq:approx_offdiag} as
\begin{eqnarray*}
\int_\Rdd a(z) \mu^{N}(\psi,\phi)(z) \dd z &=& 
\int_\Rdd a(z) \sum_{j=0}^{N-1} (-1)^j C_{N-1,j} \sum_{\substack{k\in \N^d \\ |k|=j}} \W_{\ph_k}* \W(\psi,\phi)(z)~ \dd z \\
&=& \int_\Rdd a(z) \sum_{j=0}^{N-1} (-1)^j C_{N-1,j} \sum_{\substack{k\in \N^d \\ |k|=j}} 
V_{\ph_k}\psi(z) \; \overline{V_{\ph_k}\phi(z)}~ \dd z \\
&=& \sum_{j=0}^{N-1} (-1)^j C_{N-1,j} \sum_{\substack{k\in \N^d \\ |k|=j}} \lw V_{\ph_k}^* a V_{\ph_k}\psi, \phi \rw_{L^2(\Rd)} \\
&=& \sum_{j=0}^{N-1} (-1)^j C_{N-1,j} \sum_{\substack{k\in \N^d \\ |k|=j}} \lw \opaw^{\ph_k}(a) \psi, \phi \rw_{L^2(\Rd)}.
\end{eqnarray*}
Hence, the spectrogram approximation from
Theorem~\ref{thm:spec_exp}
can be rewritten in the following operator form:
\begin{prop}\label{prop:weyl_expansion_via_antiwick}
Let  $N\in \N$ an $\hbar>0$. Then, for all $\hbar$-independent 
Schwartz class functions $a:\Rdd \to \C$ it holds
\[
\op(a) = \sum_{j=0}^{N-1} (-1)^j C_{N-1,j} \sum_{\substack{k\in \N^d \\ |k|=j}} \opaw^{\ph_k}(a) + O(\hbar^N)
\]
in the operator norm topology on $L^2(\Rd)$.
\end{prop}

One can generalize Proposition~\ref{prop:weyl_expansion_via_antiwick} in the usual sense by allowing for
more general symbol classes. In particular, Proposition~\ref{prop:weyl_expansion_via_antiwick}
remains true as long as $a$ belongs to a suitable Shubin class 
$\Gamma_{\rho}^{2N(1-\rho)}(\Rdd)$ of symbols, where
\begin{equation}\label{eq:shubin_class}
\Gamma^m_\rho(\Rdd) = \Big\{ a\in \C^\infty(\Rdd,\C): | \partial_z^\alpha a(z)|\leq  
C_\alpha \lw z\rw^{m-\rho|\alpha|} \;\;\forall
z\in\Rdd,\;\; \alpha \in \N^{2d} \Big\}
\end{equation}
with $\lw z \rw = (1+|z|^2)^{1/2}$.
Note that the Weyl quantization of a symbol $a\in \Gamma^m_\rho(\Rdd)$ creates a
bounded operator from the Shubin-Sobolev space
\[
Q^m(\Rd) = \Big\{  \psi\in \mathcal{S}'(\Rd):  (1+|x|^2-\Delta)^{-m/2}\psi \in L^2(\Rd) \Big\}
\]  
into $L^2(\Rd)$, and it is known that $Q^m(\Rd)$ actually coincides (with equivalent norm) with the 
modulation space $M^2_{\lw z\rw^m}(\Rd)$, see~\cite{BCG04,L11}. 
For example, a more general version of Proposition~\ref{prop:weyl_expansion_via_antiwick}
can be formulated as follows.

\begin{cor}\label{cor:spec_approximation_general}
Let  $N\in \N$, $\hbar>0$ and assume $a\in \Gamma^m_\rho(\Rdd)$ for $m\in \R,\rho\geq 0$. 
Then, 
\begin{eqnarray*}
\op(a) - \sum_{j=0}^{N-1} (-1)^j C_{N-1,j} \sum_{\substack{k\in \N^d \\ |k|=j}} \opaw^{\ph_k}(a) =  O(\hbar^N) 
\end{eqnarray*}
as a bounded operator from $M^2_{m-2N\rho}(\Rd)$ into $L^2(\Rd)$.
\end{cor}

\subsection{Polyanalytic Bargmann representation of antiholomorphic Weyl quantized polynomials}

The close connection between polyanalytic Bargmann transforms and short-time Fourier
transforms with Hermite windows allows to rephrase 
Proposition~\ref{prop:weyl_expansion_via_antiwick} in a Fock space setting.
In particular, the important property of almost-invariance of polyanalytic Fock spaces
under multiplication with holomorphic polynomials allows to prove the 
following result that might allow new insights about the manipulation of signals in
a multiplexing setup, see~\cite{AG10,Ab10}.

\begin{prop}[Polyanalytic Bargmann representation of  antiholomorphic 
Weyl operators]\label{prop:antiholo_bargmann}
Let  $N\in \N$, $\hbar>0$ and $p:\C^d \to \C$ be a 
(holomorphic) polynomial of degree $N-1$. Then, one has
\[
\op(\breve p) = \sum_{j=0}^{N-1} (-1)^j C_{N-1,j} \emph\BB_j^*  p \emph\BB_j 
\]
where $\breve p(q,p) := p(q-ip)$ is the transformation 
from~\eqref{eq:flip_imag_part} and $\emph\BB_j=\sum_{|k|=1}\B_k$.
\end{prop}

\begin{proof}
We start by rewriting the generalized anti-Wick operators
in the anti-Wick expansion from Proposition~\ref{prop:weyl_expansion_via_antiwick}
in terms of polyanalytic Bargmann transforms,
\begin{eqnarray*}
\op(\breve p)  = \sum_{j=0}^{N-1} (-1)^j C_{N-1,j} \sum_{\substack{k\in \N^d \\ |k|=j}} \opaw^{\ph_k}(\breve p) 
&=&
\sum_{j=0}^{N-1} (-1)^j C_{N-1,j} \sum_{\substack{k\in \N^d \\ |k|=j}} \B_k^*p \B_k,
\end{eqnarray*}
where the error vanishes because $p$ is of sufficiently low degree, see Theorem~\ref{thm:spec_exp}.
Since $p$ is a holomorphic polynomial, for each true polyanalytic Fock space $\F^k(\Cd)$
 multiplication by $p$ leaves a dense subset of $\F^k(\Cd)$ consisting of true polyanalytic polynomials
invariant. Moreover, true polyanalytic Fock spaces are othogonal: 
for any  $u\in \F^k(\Cd)$ and $v\in \F^\ell(\Cd)$
with $k\neq \ell$ it holds
\[
\lw u,v\rw_{L^2_\Phi} = 0.
\]
Thus, the result follows from observing
\begin{eqnarray*}
\sum_{j=0}^{N-1} (-1)^j C_{N-1,j} \sum_{\substack{k\in \N^d \\ |k|=j}} \B_k^*p \B_k
&=& \sum_{j=0}^{N-1} (-1)^j C_{N-1,j} \BB_j^*  p \BB_j  .
\end{eqnarray*}

\end{proof}

We can revisit this result in the context
of multiplexing as e.g. considered in~\cite{AG10}. Namely,  polyanalytic
Bargmann transforms allow to transform $n$ signals $(\psi_0,\hdots,\psi_{n-1})$
into the single signal
\[
\B \psi_0 + \BB_1 \psi_1 + \hdots + \BB_{n-1} \psi_{n-1}: \C^d \to \C
\]
that now can jointly be transmitted or manipulated. Afterwards, 
the $n$ original signals can be recovered via orthogonal projection by 
using the suitable (polyanalytic) Bergman projectors. This
is an implication of the orthogonality of polyanalytic Fock spaces of 
different degree.

In other words, Proposition~\ref{prop:antiholo_bargmann} can be understood
in the sense that the polynomial 
manipulation of a single multiplexed signal with arbitrary 
number of ``multiplexing copies'' can be expressed in terms of the action
of usual Weyl operators. For more general manipulations the error terms
from the spectrogram expansion can be used when approximating
the multi-level Bargmann multiplier by a Weyl operator.

\subsection{A polyanalytic Toeplitz expansion of complex Weyl operators}

The aim of this section is to provide a version of the anti-Wick expansion from 
Proposition~\ref{prop:weyl_expansion_via_antiwick} in the complex setting.
That is, instead of anti-Wick operators we employ the earlier defined polyanalytic 
Toeplitz operators and relate them to complex Weyl operators as considered in~\cite[\S13]{Z12}. 

Let us recall the holomorphic quadratic phase~$\theta$ from~\eqref{eq:theta_phase}
that charaterizes the Bargmann transform. In fact, $\theta$ gives rise to the 
complex symplectic map
\begin{equation}\label{eq:kappa_map}
\kappa:\C^{2d} \to \C^{2d},\quad \kappa(z,w)\mapsto (\i w-z,\tfrac12(\i z - w))
\end{equation}
by means of the implicit generating function type definition
\[
\kappa(w,-\partial_w \theta(z,w) ) = (z, \partial_z \theta(z,w)), \quad z,w \in \C^d.
\]
One can show that $\kappa$ is a bijection as a map from $\R^{2d}$ on the Lagrangian subspace
\begin{equation}
\Lambda = \{ (z,-\tfrac{i}2  \overline{z}): z\in \C^d\} \subset \C^{2d}
\end{equation}
of real dimension $2d$. The subspace $\Lambda$ is Im-Lagrangian and Re-symplectic with respect to the 
complex symplectic form $\sigma_\C = \sum_{j=1}^d \dd w\wedge \dd z$ on $\C^n\times \C^n$,
that is,
\[
\Im \sigma_\C \restriction_{\Lambda} = 0 \quad \text{ and }\quad \Re \sigma_\C \restriction_{\Lambda} \text{  is nondegenerate}.
\]
In particular~$\Lambda $ is only $\R$-linear but not $\C$-linear and, hence, is not
of the type of complex Lagrangian subspaces usually considered in the
parametrization of generalized coherent states, see, e.g.,~\cite{DKT17}. 
In other words, $\Lambda$ is an isotropic subspace of maximal dimension, but the Hermitian
form
\[
\C^{2d}\ni z \mapsto \tfrac\i2  \lw z, \Omega z\rw_{\C^{2d}} = \tfrac\i2 \overline z \cdot \Omega z 
\]
with the standard symplectic matrix
\[
\Omega = \begin{pmatrix}
0 & -\Id\\ \Id & 0
\end{pmatrix}
\]
is neither positive nor negative definite on $\Lambda$, since one computes
\[
\lw z, \Omega z\rw_{\C^{2d}}  = \tfrac12 (\Im(\zeta)^2 - \Re(\zeta)^2)\quad 
\forall z = (\zeta ,-\tfrac\i2\overline \zeta) \in \Lambda.
\]

In~\cite[\S13]{Z12} the symplectic mapping $ \kappa$ from \eqref{eq:kappa_map} is used to  introduce a
complex Weyl quantization on the Bargmann transform side. Namely, the bijection $\kappa$  can be
used to identify $\Cd$ with the Lagrangian 
subspace $\Lambda \subset \C^{2d}$ and for a Schwartz function 
$a:\Lambda \to \C$ we define its Weyl quantization 
\begin{equation}
\op_\Phi(a):L^2_\Phi(\Cd) \to L^2_\Phi(\Cd) 
\end{equation}
via the usual Fourier integral formalism
\begin{equation}
\op_\Phi(a) f(z) = (2\pi\hbar)^{-d} \int_{\Gamma_\Phi(z)} a(\tfrac{z+w}2) \e^{\i (z-w)\zeta	/\hbar} f(w) \dd \zeta \wedge \dd w
\end{equation}
along the $z$-dependent contour 
\[
\Gamma_\Phi(z): w \mapsto \tfrac{2}\i \partial_z \Phi(\tfrac{z+w}2) = -\tfrac{i}2  \overline{\tfrac{w+z}2}.
\]
One can check that $\op_\Phi(a)$ defines a bounded operator
both on $L^2_\Phi(\Cd)$ and the Fock space $\F(\Cd)$. Now, the Bargmann transform 
appears as the appropriate translation between
real and complex Weyl quantization.

\begin{lem}[see Theorem 13.9 from \cite{Z12}]\label{lem:connection_weyl_real_comp}
For any Schwartz function $a:\Lambda \to \C$ one has
\[
\B^* \op_\Phi(a) \B = \op(\kappa^\star a)
\]
where  $\kappa^\star$ denotes the pull-back by $\kappa$.
\end{lem}
Note that Lemma~\ref{lem:connection_weyl_real_comp} naturally extends to larger symbol classes, 
in particular to Shubin classes $\Gamma^m_\rho(\Lambda)$ that consist of functions $a$ for 
which $\kappa^\star a \in \Gamma^m_\rho(\Rdd)$, see also~\eqref{eq:shubin_class}. 
We apply Lemma~\ref{lem:connection_weyl_real_comp}
to obtain an expansion of complex Weyl quantized operators
in terms of $k,0$-polyanalytic Toeplitz operators and, thus, provide a complex version of
Proposition~\ref{prop:weyl_expansion_via_antiwick}.

\begin{thm}
Let  $N\in \N$, $\hbar>0$ and assume  
$  a \in \Gamma^m_\rho(\Lambda)$ for $m\in \R,\rho\geq 0$. 
Then, one has the approximation
\begin{eqnarray*}
\Big\|  \op_\Phi ( a)  - \sum_{j=0}^{N-1} (-1)^j C_{N-1,j} 
\sum_{\substack{k\in \N^d \\ |k|=j}} \T_{k,0}(\widehat  {\kappa^\star  a}) 
\Big\|_{\F^{m-2N\rho}(\Cd) \to \F(\Cd)} 
= O(\hbar^N).
\end{eqnarray*}
where  $\widehat{\kappa^\star  a}(z) = \kappa^\star  a(q,-p)$ with  $z=q+ i p \in \Cd$.
\end{thm}

\begin{proof}
We have to show that $\op_\Phi (a) :\F^{m-2N\rho}(\Cd) \to \F(\Cd)$ which
then implies $\P \op_\Phi (a) \P = \op_\Phi (a)$ as an operator on $\F^{m-2N\rho}(\Cd)$.
For any $u \in \F_{m-2N\rho}(\Cd)$, we can rewrite
\[
\op_\Phi (a) u (z) = \int_\Cd K_a(w,z) u(w) \dd w
\]
with the Schwartz kernel
\[
K_a(w,z) = (2\pi\hbar)^{-d} a(\tfrac{z+w}2) \e^{(z-w)\overline{(z+w)}/4\hbar} f(w)
\]
where 
\[
\frac{\dd}{\dd \overline z}  K_a(w,z) = \frac{\dd}{\dd \overline w}  K_a(w,z).
\]
Then, the holomorphy of $\op_\Phi (a) u $ follows directly by the holomorphy of $u$ since
\begin{eqnarray*}
\frac{\dd}{\dd \overline z} \op_\Phi (a) u (z) &=& \frac{\dd}{\dd \overline z} \int_\Cd K_a(w,z) u(w) \dd w \\
&=& \int_\Cd \frac{\dd}{\dd \overline w} K_a(w,z)u(w) \dd w \\
&=&\int_\Cd  K_a(w,z)( - \tfrac{\dd}{\dd \overline w} u)(w) \dd w =0.
\end{eqnarray*}
The appropriate decay of $\op_\Phi (a) u$ can be inferred from the intertwining property in Lemma~\ref{lem:connection_weyl_real_comp}
and the maping properties of usual Weyl quantized operators on modulation spaces.
Finally, the approximation order follows  from Corollary~\ref{cor:spec_approximation_general}.
\end{proof}

\section{Outlook}

The concept of polyanalytic Toeplitz operators we propose in this paper appears quite 
straight-forward once written down and naturally exhibits all 
the favorable mapping qualities that are known from the analytic Bargmann setting. 
However, by the connection to short-time Fourier transforms and, via
 the spectrogram expansion, to Weyl operators this new concept allows to formulate profound
transition and approximation formulas for the whole range of real and complex Weyl, Toeplitz as well as localization
 operators.

We believe that polyanalytic Toeplitz quantization might prove a useful concept in 
a variety of areas, including the deeper investigation and approximation of multiplexed 
signals, the analysis and generalization of complex quantization theories and a geometrically
satisfying complex generalization of coherent state approximations and dynamics.

\begin{appendix}
\section{Moments of special Hermite functions}\label{app_moments_wigner_hermite}

The Wigner transforms $\W_{\ph_k}$ of Hermite functions are also known as special Hermite functions, see, e.g., \cite{Th93}. 
Moments of these functions are of special interest as they resemble the quantum expectation values of quantized monomials in 
the $k$th harmonic oscillator eigenstate. That is,
\[
\int_\Rdd z^\alpha \W_{\ph_k}(z) \dd z =  \lw\ph_k,\op(z^\alpha)\ph_k \rw
\]
with standard multiindex notation, where $\a\in \N^{2d}$. As the Wignerfunctions of mulitdimensional Hermite functions factorize 
into $2$-dimensional Wigner functions, in the following we only compute formulas for this case by proving Proposition~\ref{prop:2d_hermite_moments}.

\begin{proof}[Proof of Proposition~\ref{prop:2d_hermite_moments}]
We start by computing
\begin{align*}
\int_{\R^2}& x^{2\alpha} \xi^{2\beta }\W_{\ph_k} (x,\xi) \dd x \dd\xi 
= (\pi \hbar)^{-1} \int_{\R^2}  x^{2\alpha} \xi^{2\beta } \e^{-(x^2+\xi^2)/\hbar} (-1)^{k} L_{k}(\tfrac{2}{\hbar}(x^2+\xi^2)) \dd x \dd\xi  \\
&= \pi^{-1} \hbar^{\a+\beta} \int_{\R^2}  x^{2\alpha} \xi^{2\beta } \e^{-(x^2+\xi^2)^2} (-1)^{k} L_{k}(2(x^2+\xi^2)) \dd x \dd\xi \\
&= \pi^{-1} \hbar^{\a+\beta} \int_{\R^2}  x^{2\alpha} \xi^{2\beta } \e^{-(x^2+\xi^2)^2} (-1)^{k}  \sum_{j=0}^k {k \choose k-j} \frac{(-2(x^2+\xi^2))^j}{j!} \dd x \dd\xi = (*)
\end{align*}
by using the Laguerre formula~\eqref{eq:wigner_hermite} for the Wigner function $\W_{\ph_k}$. Expanding the polynomial in $x$ and $\xi$ 
by the binomial theorem  and using the definition of the Gamma function we get
\begin{align*}
(*)  &= \pi^{-1} \hbar^{\a+\beta}(-1)^{k}  \sum_{j=0}^k \frac{(-2)^j}{j!} {k \choose k-j} \sum_{n=0}^j  {j \choose n} \int_{\R^2}  \e^{-x^2-\xi^2} {x^{2(n+\a)}\xi^{2(j-n+\beta)}} \dd x \dd \xi \\
&= \pi^{-1} \hbar^{\a+\beta}(-1)^{k}  \sum_{j=0}^k \frac{(-2)^j}{j!} {k \choose k-j} \sum_{n=0}^j  {j \choose n} \Gamma(\tfrac12 + n+\a) \Gamma(\tfrac12 + j-n+\beta) .
\end{align*}
Finally, using binomial sum theorems for Gamma functions and the hypergeometric function $ _2{\text{F}}_1$ we compute
\begin{align*}
(*)  &= \pi^{-1} \hbar^{\a+\beta}(-1)^{k} \frac{\Gamma(\a+\tfrac12)\Gamma(\beta+\tfrac12)}{\Gamma(\a+\beta+1)} \sum_{j=0}^k \frac{(-2)^j}{j!} {k \choose k-j}  \Gamma(\a+\beta+1+j)  \\
&= \pi^{-1} \hbar^{\a+\beta}(-1)^{k} \Gamma(\a+\tfrac12)\Gamma(\beta+\tfrac12)  _2{\text F}_1(\a+\beta+1,-k;1;2)  \\
&=  \hbar^{\a+\beta}(-1)^{k} \Gamma(\a+\tfrac12)\Gamma(\beta +\tfrac12)  _2{\text F}_1(\a+\beta+1,-k;1;2)  \\
&=  { } _2{\text F}_1(\a+\beta+1,-k;1;2)(-1)^{k}   \frac{\hbar^{|\alpha|} (2\a)!}{4^{|\a|}\a!}
\end{align*}
which completes the proof.
\end{proof}

\end{appendix}

\providecommand{\noopsort}[1]{}\providecommand{\singleletter}[1]{#1}%


\end{document}